\documentclass[a4paper,11pt]{article}
\usepackage[margin=3cm]{geometry}
\usepackage{paralist}
\usepackage{hyperref}
\usepackage{amssymb, amsmath, amsthm}
\usepackage{graphicx}
\usepackage{enumitem}
\usepackage[ruled,lined,algo2e]{algorithm2e}
\usepackage[font=small,labelfont=bf]{caption}
\usepackage[font=small,labelfont=normalfont]{subcaption}
\usepackage{times}
\usepackage{wrapfig}
\usepackage{xspace}
\usepackage{todonotes}

\newenvironment{proofMath}{\par\addvspace\topsep\noindent{\bf Proof:} \ignorespaces }{}

\title{On Monotone Drawings of Trees\thanks{A preliminary version of this work 
	appeared in: Proc.\ 22nd International Symposium on Graph Drawing
	(GD'14)~\cite{kssw-omdt-14}. This research was
    supported by the ESF EuroGIGA project GraDR (DFG grant
    Wo~758/5-1).}}

\author{%
  Philipp~Kindermann\thanks{Lehrstuhl f\"ur Informatik I,
  Universit\"at W\"urzburg, Germany.
  \mbox{http://www1.informatik.uni-wuerzburg.de/en/staff}}
  \and
  Andr{\'e} Schulz\thanks{Institut f\"ur Mathematische Logik und
  Grundlagenforschung, Universit\"at M\"unster, Germany.
  \mbox{andre.schulz@uni-muenster.de}}
  \and
  Joachim Spoerhase\footnotemark[2]
  \and
  Alexander~Wolff\footnotemark[2]
}

\date{}

\theoremstyle{plain}
\newtheorem{theorem}{Theorem} 
\newtheorem{lemma}{Lemma}

\graphicspath{{figs/}}

\begin{document}
\maketitle

\begin{abstract}
  A crossing-free straight-line drawing of a graph is \emph{monotone} if
  there is a monotone path between any pair of vertices with respect
  to \emph{some} direction.  We show how to construct a monotone drawing 
  of a tree with $n$ vertices on an $O(n^{1.5}) \times O(n^{1.5})$ grid
  whose angles are close to the best possible angular resolution.  Our
  drawings are \emph{convex}, that is, if every edge to a leaf is
  substituted by a ray, the (unbounded) faces form convex regions.
  It is known that convex drawings are monotone and, in the case of
  trees, also crossing-free.

  A monotone drawing is \emph{strongly monotone} if, for every pair of
  vertices, the direction that witnesses the monotonicity comes from the
  vector that connects the two vertices.  We show that every tree admits
  a strongly monotone drawing.  For biconnected outerplanar graphs,
  this is easy to see.  On the other hand, we present a
  simply-connected graph that does not have a strongly monotone
  drawing in any embedding.
\end{abstract}

\section{Introduction}

A natural requirement for the layout of a connected graph is that
between any source vertex and any target vertex, there should be a
source--target path that approaches the target according to some
distance measure.  A large body of literature deals with problems of
this type; various measures have been studied.  For example, in a
\emph{greedy drawing} you can find a path to a target vertex by iteratively
selecting a neighbor that is closer to the target.  
In a \emph{monotone} drawing, the distance between vertices (on the desired
source-target path) is measured with respect to their projections on
\emph{some} line, which may be different for any source--target pair.
We say that a path~$P$ is \emph{monotone with respect to a vector~$\vec v$}
if the orthogonal projection of the vertices of~$P$ on every line with direction 
vector~$\vec v$ appears on the line in the order as induced by~$P$.
We also refer to~$\vec v$ as a \emph{direction}.
In \emph{strongly monotone} drawings, that line is always the line
from source to target, and in \emph{upward} drawings, the line is
always the vertical line, directed upwards.

In this paper, we focus on monotone and strongly monotone drawings of
trees with additional aesthetic properties such as convexity or small
area.  Given a tree, we call the 
edges incident to the leaves \emph{leaf edges} and all other edges
\emph{interior edges}.  Given a straight-line drawing of a tree, we
substitute each leaf edge by a ray whose initial part coincides with
the edge. The embedding of the tree in the plane defines a combinatorial embedding
of the tree, that is, the order of the edges around every vertex. The
faces are then specified by this combinatorial embedding as leaf--leaf paths. 
If the faces of the augmented drawing are realized as convex nonoverlapping
(unbounded) polygonal regions, then we call the original drawing a
\emph{convex drawing}.  If every region is \emph{strictly convex}
(that is, all interior angles are strictly less than~$\pi$), we also
call the drawing \emph{strictly convex}.  Note that a strictly convex
drawing is also monotone~\cite{acbfp-mdg-12,acm-mpaoa-SoCG89}, but a
monotone drawing is not necessarily convex.  Strict convexity forbids vertices 
of degree~2. When we talk about (strongly) monotone drawings, 
this always includes the planarity requirement.  Otherwise, as Angelini et
al.~\cite{acbfp-mdg-12} observed, drawing any spanning tree of the given
graph in a (strongly) monotone way and inserting the remaining edges
would yield a (strongly) monotone drawing of the graph.

\paragraph{Previous Work.}
Rao et al.~\cite{rpss-grwli-MOBICOM03} introduced the concept of 
greedy drawings for a coordinate-based routing algorithm that does not rely
on location information. While any 3-connected plane graph has a greedy drawing
in the Euclidean plane~\cite{lm-gems-10} (even without crossing~\cite{d-gdt-10}),
this is, unfortunately, not true for trees.  
N\"ollenburg and Prutkin~\cite{np-egdt-ESA13} gave a complete
characterization for the tree case, which shows that no tree with a
vertex of degree~6 or more admits a greedy drawing.  

Alamdari et al.~\cite{acglp-sag-GD12} studied a subclass of greedy drawings,
so-called \emph{self-approaching drawings} which require that there
always is a source--target path such that the distance decreases for
any triplet of intermediate points on the \emph{edges}, not only for
the vertices on the path.  These drawings are based on the concept of
self-approaching curves~\cite{ikl-sac-MPCPS95}. 

Carlson and Eppstein~\cite{ce-tcfoa-GD06} studied convex drawings of
trees.  They give linear-time algorithms that optimize the angular
resolution of the drawings, both for the fixed- and the
variable-embedding case.  They observe that convexity allows them to
pick edge lengths arbitrarily, without introducing crossings.

For monotone drawings, Angelini et al.~\cite{acbfp-mdg-12} studied the
variable-embed\-ding case.  They showed that any $n$-vertex tree admits
a straight-line monotone drawing on a grid of size $O(n^{1.6}) \times
O(n^{1.6})$ (using a BFS-based algorithm) or $O(n) \times O(n^2)$
(using a DFS-based algorithm).  They also showed that any biconnected
planar graph has a monotone drawing (using exponential area).
Further, they observed that not every planar graph admits a monotone
drawing if its embedding is fixed.  They introduced the concept of
\emph{strong monotonicity} and showed that there is a drawing of a
planar triangulation that is not strongly monotone.  Hossain and
Rahman~\cite{hr-mgdpg-FAW14} improve some of the results of Angelini
et al.\ by showing that every connected planar graph
admits a monotone drawing of size $O(n) \times O(n^2)$ and that such a
drawing can be computed in linear time. 

Both the BFS- and the DFS-based algorithms of Angelini et al.\
precompute a set of $n-1$ integral vectors in decreasing order 
of slope by using two different partial traversals of the
so-called \emph{Stern--Brocot tree}, an infinite tree whose vertices are in
bijection with the irreducible positive rational numbers.  Such
numbers can be seen as \emph{primitive} vectors in 2D, that is,
vectors with pairwise different slopes.  Then, both algorithms do a
pre-order traversal of the input tree.  Whenever they
hit a new edge, they assign to it the steepest unused vector.  They
place the root of the input tree at the origin and draw each
edge~$(u,v)$ by adding its assigned vector to the position
of~$u$.  They call such tree drawings \emph{slope-disjoint}.  We will not
formally define this notion here, but it is not hard to see that it
implies monotonicity.  

Angelini, with a different set of co-authors~\cite{adkmrsw-mdgfe-A15},
investigated the fixed-embedding case.  They showed that, on the
$O(n) \times O(n^2)$ grid, every connected plane graph admits a
monotone drawing with two bends per edge and any outerplane graph
admits a straight-line monotone drawing.

\paragraph{Our contribution.} 
We present two main results.  First, we show that any $n$-vertex tree
admits a convex and, hence, monotone drawing on the $O(n^{1.5}) \times
O(n^{1.5})$ grid (see Section~\ref{sec:griddrawing}).  As the drawings of
Angelini et al.~\cite{acbfp-mdg-12}, our drawings are slope-disjoint, but
we use a different set of primitive vectors (based on Farey sequences),
which slightly decreases the grid size and helps us achieve better angular resolution.  (This also works
for the BFS-based algorithm of Angelini et al.)  Instead of pre-order,
we use a kind of in-order traversal (first child~-- root~-- other 
children) of the input tree, which helps us to achieve convexity.
Our ideas can be applied to modify the (non-grid) optimal angular resolution 
algorithm of Carlson and Eppstein~\cite{ce-tcfoa-GD06} such that
a drawing on an $O(n^{1.5}) \times
O(n^{1.5})$ grid is constructed at the expense of missing the optimal
angular resolution by a constant factor. 

Second, we show that any tree admits a \emph{strongly} monotone
drawing (see Section~\ref{sec:strong}).  So far, no positive results
have been known for strongly monotone drawings.
In the case of proper binary trees, 
our drawings are additionally strictly convex.  For biconnected outerplanar
graphs, it is easy to construct strongly monotone drawings.  On the
other hand, we present a simply-connected planar
graph that does not have a strongly monotone drawing in any embedding.

\paragraph{Subsequent Work.} Subsequent to our work, He and He~\cite{hh-cmdt-COCOON15} 
improved the area bound for monotone drawings of trees to 
$O(n^{1.205}) \times O(n^{1.205})$.  Other than the drawings of
  our algorithm, their drawings are not necessarily convex.  Their
algorithm follows our approach  
of using Farey sequences to acquire a set of primitive vectors and then
computing a slope-disjoint
drawing. Recently, the same authors \cite{hh-nomdot-TCS16}
  further reduced the area bound to $O(n\log n)\times O(n\log n)$.
For triconnected planar graphs, He and He~\cite{hh-md3cp-ESA15} proved
that the convex drawings that an algorithm of
Felsner~\cite{f-cdpgo-O01} places on a grid of size $O(n) \times O(n)$
are even monotone.  Note, however, that augmenting a graph that is not triconnected to triconnectivity, running Felsner's algorithm and then
removing the additional edges will, in general, neither yield a convex
nor a monotone drawing.  Hence, the result for triconnected graphs
does not imply any improvement for trees.
Very recently, Felsner et al.~\cite{fikkms-smdpg-socg16} showed that
all 3-connected planar graphs, outerplanar graphs, and 2-trees admit a strongly
monotone drawing. Their algorithm for outerplanar graphs utilizes an alternate
proof that every tree admits a strongly monotone drawing, but they also achieve
convexity.

\section{Building Blocks: Primitive Vectors}
\label{sec:primitive}

The following algorithms require a set of integral vectors with distinct 
directed slopes and bounded length. In particular, we ask for a set of 
\emph{primitive vectors} $P_{d}=\{ (x,y) \mid \gcd(x,y) \in\{1,d\}, 0 \le x \le y \le d \}$.
Our goal is to find the right value of $d$ such that $P_{d}$ contains
at least $k$ primitive vectors, where $k$ is a number that we determine later. 
We can then use the reflections on the lines $x=y$, $y=0$ and $x=0$ to
get a sufficiently large set of integer vectors with  distinct
directed slopes.  The edges of the monotone drawings in
Section~\ref{sec:griddrawing} are translates of these vectors; each edge 
uses a different vector. 

\begin{figure}[tb]
  \centering
  \includegraphics{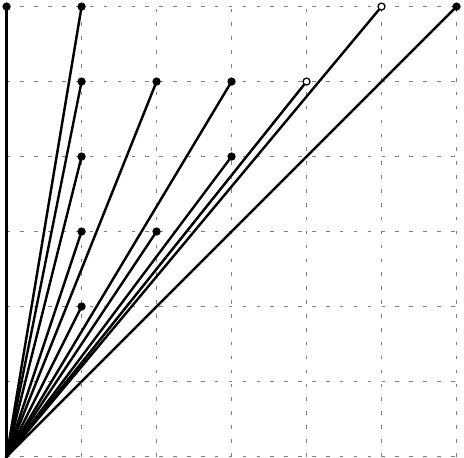}
  \caption{The 13 primitive vectors obtained from~$\mathcal{F}_6$. The
    smallest angle of~$\approx$~$1.14^\circ$ is realized between the
    vectors $(4,5)$ and $(5,6)$ marked with white dots; the best
    possible angular resolution in this case is $45^\circ/12 =
    3.75^\circ$. Note that our algorithm would use $\mathcal{F}_{12}$ to
    acquire~13 primitive vectors.}
  \label{fig:primitive}
\end{figure}

Assume that we have fixed $d$ and want to generate the set~$P_d$. If we
consider each entry $(x,y)$ of~$P_d$ to be a rational number $x/y$ and
order these numbers by value, we get the \emph{Farey sequence}~$\mathcal{F}_d$ 
(see, for example, Hardy and Wright's book~\cite{gw-itn-79}). 
The Farey sequence is well understood. In particular, it is known that
$|\mathcal{F}_d| = 3d^2/\pi^2 + O(d \log d)$~\cite[Theorem~331]{gw-itn-79}. 
Furthermore, the entries of~$\mathcal{F}_d$ can 
be computed in time $O(|\mathcal{F}_d|)$.  We remark that the set~$\bigcup_d \mathcal{F}_d$ coincides with the entries of the
Stern--Brocot tree.  However, collecting the latter level by level is
not the most effective method to build a set of primitive vectors for
our purpose. 

To obtain a set of~$k$ primitive vectors, we use the first~$k$
entries of the Farey sequence~$\mathcal{F}_d$, for~$d:= \lceil 
\pi^2 \sqrt k/3  \rceil$, 
replacing each rational by its corresponding two-dimensional vector.  
We select exactly~$k$ primitive vectors from this set which we denote by~$V_k$;
see Figure~\ref{fig:primitive}.

If we wish to have more control over the aspect ratio in our final drawing, we 
can pick a set of primitive vectors contained inside a triangle spanned
by the grid points~$(0,0)$, $(m_x,0), (m_x,m_y)$. By stretching the triangle and 
keeping its area fixed, we may end up with fewer primitive vectors.
This will result in an (only slightly) smaller constant compared to
the case $m_x=m_y$.  
As proven by B\'ar\'any and Rote~\cite[Theorem~2]{br-scdpg-06}, any such
triangular domain contains at least~$m_x m_y / 4$ primitive
vectors. This implies that we can adapt
the algorithm easily to control the
aspect ratio by selecting the box for the primitive vectors 
accordingly. For the sake of simplicity, we detail our algorithms only
for the most interesting case ($m_x=m_y$).

\begin{lemma}\label{lem:primitive_angles}
  Let $P\subseteq P_d$ be a set of $k=|P_d|/c$ primitive vectors with
  no coordinate greater than~$d$ for some constant~$c\ge 1$.  Then, any
  two primitive vectors of $P$ are separated by an angle of~$\Omega(1/k)$.
\end{lemma}

\begin{proofMath}
Since~$|P_d|=3d^2/\pi^2 + O(d \log d)$, we have that~$2d^2 \approx 2\pi^2 ck/3$.
Any line with slope~$m$ encloses an angle~$\alpha$ with the $x$-axis, 
such that~$\tan(\alpha)=m$. Let~$m_1$ and~$m_2$ be the slopes of two lines 
and let~$\alpha_1$ and~$\alpha_2$ be the corresponding angles with respect to 
the $x$-axis. By the trigonometric addition formulas we have that the 
separating angle~$\phi$ of these two lines is such that:
$$\tan \phi := \tan(\alpha_1-\alpha_2) 
= \frac{\tan\alpha_1-\tan\alpha_2}{1+ \tan \alpha_1 \tan \alpha_2} 
= \frac{m_1 - m_2}{1+m_1m_2}.$$ 
For any two neighboring entries~$p/q$ and~$r/s$ in 
the Farey sequence, it holds that $qr-ps=1$ \cite[Theorem~3.1.2]{gw-itn-79}, and 
therefore~$p/q$ and~$r/s$  differ by exactly~$(qr-ps)/(qs)=1/(qs)$.  Now assume that~$\phi$ is the angle 
between the vectors~$(p,q)$ and~$(r,s)$.  As a 
consequence, $\tan \phi = 1/(pr+qs)$. Then,~$\phi$ is minimized if~$pr+qs$
is maximized. Clearly, we have that~$pr+qs\le 2d^2\approx 2 \pi^2c k /3$. By 
the Taylor expansion, $\arctan(x)= x - x^2 \xi /(1+\xi^2)^2$ for some 
value~$0 \le \xi \le x$. Substituting~$x$ with~$3 /(2\pi^2 c k)$ yields, 
for~$k\ge2$, that
\[\phi \ge \frac{3}{2 \pi^2c k} - \frac{9 \xi}{4 \pi^4 c^2 k^2 (1+\xi^2)^2}
> \frac{3}{2 \pi^2 c k} - \frac{9}{ 4 \pi^4 c^2 k^2} \in \Omega(1/k).
  \tag*{$\Box$}\]
\end{proofMath}

Since the best possible resolution for a set of~$k$ primitive vectors 
is~$2\pi/k$, Lemma~\ref{lem:primitive_angles} shows that the
resolution of our set differs from the optimum by at most a constant.
To estimate this constant, let us assume we use~$k=|P_d|$ primitive vectors
(that is,~$c=1$ in Lemma~\ref{lem:primitive_angles}). Then, the smallest 
angle~$\phi$ spanned by these vectors is, according to the proof of the previous 
lemma, at least~$3/(2 \pi^2 k) - 9/(16 \pi^4)$ for any $k>1$. This value should
be compared to $\text{opt}=\pi/(4k)$ since the primitive vectors span an angle 
of~$\pi/4$ in total. We obtain that the ratio $\phi /\text{opt}$ is smaller than~6. 

\section{Monotone Grid Drawings with Large Angles}
\label{sec:griddrawing}

In this section, we present a simple method for drawing a
tree on a grid in a strictly convex, and therefore monotone way.
Lemma~\ref{lem:crossing-free} shows that this drawing is automatically
crossing-free. We name our strategy the \emph{inorder-algorithm}.
We start by ensuring that convex tree drawings are crossing-free.
This has already been stated (without proof) by Carlson and
Eppstein~\cite{ce-tcfoa-GD06}.
\begin{lemma} 
  \label{lem:crossing-free}
  Any convex straight-line drawing of a tree is crossing-free.
\end{lemma}

  \begin{figure}[t]
    \centering
    \includegraphics{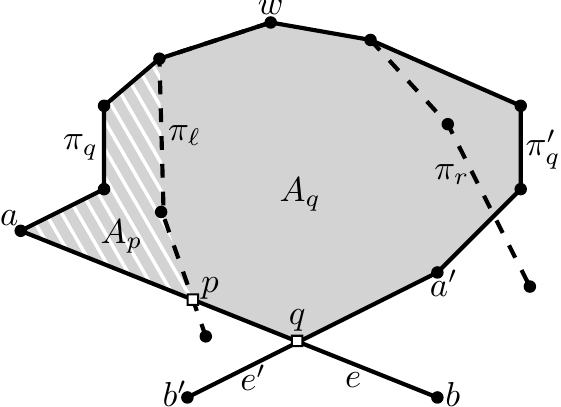}
    \caption{An illustration of the situation in the proof of
      Lemma~\ref{lem:crossing-free}} 
    \label{fig:crossingfree}
  \end{figure}

\begin{proof}
  Let $T$ be a tree and $\Gamma$ a convex straight-line drawing
  of~$T$.  Assume that two edges~$e=(a,b)$ and~$e'=(a',b')$ are
  crossing in $\Gamma$ in some point~$q$, see
  Figure~\ref{fig:crossingfree}.  Let~$w$ be the lowest common ancestor
  of~$b$ and~$b'$, let $\pi_q$ be the path~$w\to q$ via~$a$,
  and let~$\pi_q'$ the path~$w\to q$ via~$a'$. Let us assume
  that the children in~$w$ are ordered such that $\pi_q$ starts before
  $\pi_q'$.  Let~$A_q$ be the region bounded by~$\pi_q$ and~$\pi_q'$.

  We can assume that $A_q$ is of minimum area with respect to other
  crossings we may have chosen (and, hence, $A_q$ has a connected
  interior).  Now, we consider two paths starting from~$w$. The first
  one, $\pi_\ell$, starts with the first edge of $\pi_q$ and then
  always continues via the last child until it reaches a leaf. The
  second path, $\pi_r$, starts with the first edge of~$\pi_q'$ and
  continues always using the first child.  Note that the polygonal
  chain~$\pi_\ell$ together with~$\pi_r$ forms a face $f_q$ of the
  given convex drawing of the tree.  Hence, the face is convex, which
  means that~$\pi_\ell$ and~$\pi_r$ only meet in~$w$.  Furthermore, we
  either have~$\pi_\ell\not=\pi_q$ or we have~$\pi_r\not=\pi'_q$ since
  otherwise $f_q$ is self-intersecting.  As a consequence, at least
  one of the two paths, say $\pi_\ell$, enters and leaves~$A_q$.
  Let~$p$ be the point where~$\pi_\ell$ crosses~$\pi_q$ for the first
  time, and let~$A_p$ be the polygon that is bounded by the parts
  of~$\pi_q$ and~$\pi_\ell$ between~$w$ and~$p$.  Then~$A_p$ has
  smaller area than~$A_q$, which contradicts our assumption that $A_q$
  has minimum area.
\end{proof}

Our inorder-algorithm first computes a reasonable large set of primitive 
vectors, then selects a subset of these vectors, and finally assigns the slopes 
to the edges. The drawing is then generated by translating the selected 
primitive vectors. In the following, an \emph{extended}
subtree will refer to a subtree including the edge leading into the subtree
(if the subtree is not the whole tree).  

We will assign a unique number~$s(e)$ to every edge~$e$.  This number will
refer to the rank of the edge's slope (in circular order) in the final
assignment.  The rank assignment is done in a recursive fashion with increasing integral ranks from~$1$ to~$n-1$. Starting with the root, for each vertex~$v$,
 we first recursively visit its leftmost child, then assign the next rank
 to the parent edge of~$v$ (unless~$v$ is the root), and then recursively visit
 its other children from left to right. For an example of a tree with its edge 
 ranks, see Figure~\ref{fig:slopetree}.

\begin{figure}[tb]
  \begin{subfigure}[b]{.5\columnwidth}
    \centering
    \includegraphics{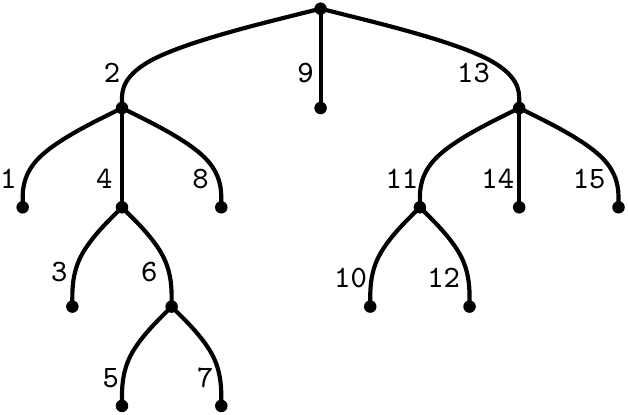}
    \caption{a tree with edge numbers~$s(\cdot)$ }
    \label{fig:slopetree}
  \end{subfigure}
  \hfill
  \begin{subfigure}[b]{.47\columnwidth}
    \centering
    \includegraphics{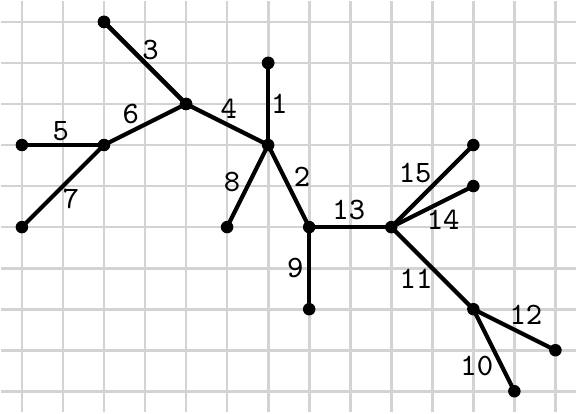}
    \caption{our grid drawing of the tree}
    \label{fig:slopetree-drawn}
  \end{subfigure}
  \caption{A strictly convex drawing of a tree}
\end{figure}

Second, we assign actual slopes to the edges.  Let~$e$ be an edge
with~$s(e)=j$.  Then, we assign some vector~$s_j\in\mathbb{Z}^2$ to~$e$ and 
draw~$e$ as a translate of~$s_j$.  We pick the
vectors~$s_1, s_2, \dots, s_{n-1}$ by selecting a sufficiently large
set of primitive vectors and their reflections in counterclockwise
order; see Section~\ref{sec:primitive}.  
Our drawing algorithm has the following requirements:

\begin{enumerate}[label=(R\arabic*),ref=R\arabic*]
\item \label{enum:root}%
  Edges that are incident to the root and consecutive in circular
  order are assigned to vectors that together span  an angle less than~$\pi$.
\item \label{enum:subtrees}%
  In every extended subtree hanging off the root, the edges 
  (including the edge incident to the root) 
  are assigned to a set of vectors that spans an angle less than~$\pi$.
\end{enumerate}

These requirements can always be fulfilled, as the following lemma shows.

\begin{lemma}\label{lem:slope-preselection}
  We can select $n-1$ vectors with distinct directed slopes 
  from a $[-d,d] \times [-d,d]$ grid with $d=4 \lceil\sqrt{n}\rceil$
  such that the requirements~\ref{enum:root}
  and~\ref{enum:subtrees} are fulfilled.
\end{lemma}

\begin{proof}
	We first preprocess our tree by adding temporary edges at some
	leaves. These edges will receive slopes, but are immediately discarded
	after the assignment.

	First, our objective is to ensure that the tree can be split up into three 
	parts that all have~$n$ edges. In particular, we adjust the sizes of the 
	extended subtrees hanging off the root by adding temporary edges such that we 
	can partition them into three sets of consecutive extended subtrees which all 
	contain~$n$ edges.  Note that we have to add~$2n+1$ edges to achieve this.

	Second, we define three cones~$C_1$, $C_2$, and~$C_3$; see 
	Figure~\ref{fig:slope-preselection}.  Each cone
	has its apex at the origin and spans an angle of~$\pi/4$.  The angular
	ranges are~$C_1=[0,\pi/4]$, $C_2=[3\pi/4,\pi]$, and $C_3=[3\pi/2,7\pi/4]$; 
	angles are measured from the $x$-axis pointing
	in positive direction.  Note that~$C_2$ is separated from the two
	other cones by an angle of~$\pi/2$.  As mentioned in
	Section~\ref{sec:primitive}, the set~$V_n$ contains~$n$ primitive
	vectors in the~$[0,d]\times[0,d]$ grid. 
	When reflected on the $x=y$ line, these vectors lie in~$C_1$. 
	Reflecting the vectors in~$C_1$, we further generate~$n$ vectors
	in~$C_2$ and~$n$ vectors in~$C_3$.  In every cone, we ``need'' at most~$n-3$ 
	edges. Hence, we can remove the vectors on the boundary of each cone.
	After removing the temporary edges, the number of vectors will drop from~$3n$ 
	to~$n-1$. 

	\begin{figure}[t]
		\centering
		\includegraphics{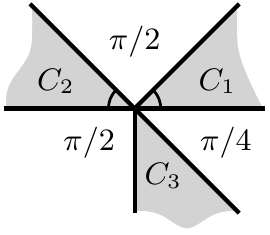}
		\caption{The cones that contain the slopes used in the algorithm}  
		\label{fig:slope-preselection}
	\end{figure}

	Now, we observe the following.  Every two consecutive edges incident to
	the root lie in the \emph{interiors} of our cones.  Given the sizes and angular 
	distances of the cones, this yields requirement~\ref{enum:root}.  
	Furthermore, any extended subtree is assigned slopes
	from a single cone.  This yields requirement~\ref{enum:subtrees}.
\end{proof}

For the example tree of Figure~\ref{fig:slopetree}, it suffices to pick
the~16 vectors that one gets from reflecting the primitive vectors
from the $[0,2] \times [0,2]$ grid.  These vectors already fulfill
requirements~\ref{enum:root} and~\ref{enum:subtrees}.  Hence, we
do not have to apply the more involved slope selection as described
in Lemma~\ref{lem:slope-preselection}.  The resulting drawing is shown
in Figure~\ref{fig:slopetree-drawn}.

Every face in the drawing contains two leaves. The leaves are ordered by 
their appearance in some DFS-sequence~$\mathcal{D}$ respecting some rooted 
combinatorial embedding of~$T$. For a face~$f$, we call the leaf that comes 
first in~$\mathcal{D}$ the \emph{left leaf} and the other leaf of~$f$  the 
\emph{right leaf} of~$f$. The only exception is the face whose leaves are the 
first and last child of~$\mathcal{D}$. Here, we call the first vertex 
in~$\mathcal{D}$ the right leaf and the last vertex in~$\mathcal{D}$ the left 
leaf.

\begin{lemma}\label{lem:inc_slopes}
  Let~$u$ be the left leaf, and let~$v$ be the right leaf of a face of~$T$. 
  Further, let~$w$ be the lowest common ancestor of~$u$ and~$v$. The above assignment 
  of slope ranks~$s$ to the tree edges implies the following.
  \begin{enumerate}[label=(a)]
		\item If edge~$e_1$ is on the path $w\to u$
			and edge~$e_2$ is on the path $w\to v$,
			then $s(e_1)<s(e_2)$.
		\item The ordered sequence of edges on the path~$w\to u$ is increasing in~$s(\cdot)$.
		\item The ordered sequence of edges on the path~$w\to v$ is decreasing in~$s(\cdot)$.
  \end{enumerate}
\end{lemma}

\begin{proof}
  Let~$a$ be an edge that links the parent~$p$ to its child~$u$, 
  let~$b$ be the edge that links~$u$ to its leftmost child, and let~$c$ be
  the edge that links~$u$ to its rightmost child;
  see~Figure~\ref{fig:slopelem2}. In the assignment, we first picked
  the slope in the subtree rooted at the leftmost children of~$u$,
  then we selected the slope for~$a$, and later we picked the slopes
  for the subtree rooted at the rightmost children of~$u$. Since we
  select the slopes in their radial order, we have $s(b)< s(a) < s(c)$.

  Now, note that the slopes on the path~$w\to u$ have been assigned
  before the slopes on the path~$w\to v$, which proves~(a).  When
  traversing the path~$w\to u$, we follow the rightmost children,
  except maybe for~$w$'s child; see~Figure~\ref{fig:slopelem1}. Hence,
  the sequence of slopes is increasing, and~(b) follows. Statement~(c)
  follows by a similar argument: We traverse the path~$w\to v$ by
  taking the leftmost child, except maybe for~$w$'s child.  Hence, the
  sequence of slopes is decreasing.
\end{proof}

  \begin{figure}[tb]
    \hfill
    \begin{subfigure}[b]{.45\columnwidth}
      \centering
      \includegraphics{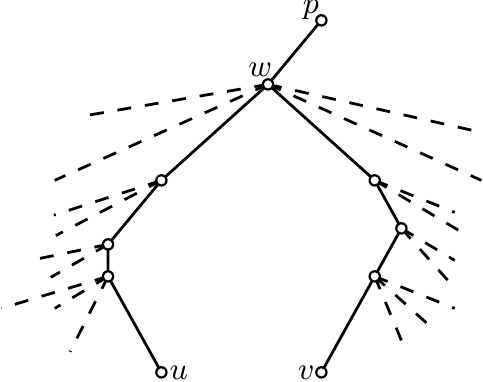}
      \caption{paths spanned by~$u,v,w$ }
      \label{fig:slopelem1}
    \end{subfigure}
    \hfill
    \begin{subfigure}[b]{.45\columnwidth}
      \centering
      \includegraphics{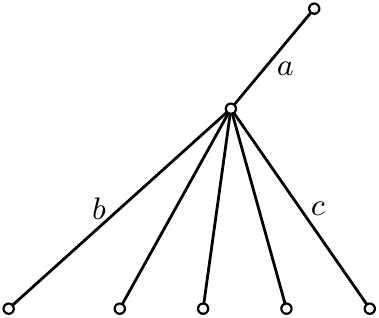}
      \caption{local assignments in the proof}
      \label{fig:slopelem2}
    \end{subfigure}
    \hfill
    \caption{Situation analyzed in Lemma~\ref{lem:inc_slopes}}
  \end{figure}

We now prove the correctness of our algorithm.

\begin{theorem}\label{thm:correctness}
  Given an embedded tree with $n$ vertices (none of degree~2), the 
	in\-order-algorithm produces a strictly convex and crossing-free drawing with
  angular resolution $\Omega(1/n)$ on a grid of size $O(n^{1.5})
  \times O(n^{1.5})$.  The algorithm runs in $O(n)$ time.
\end{theorem}
\begin{proof}
	We first show that no face in the drawing is incident to an angle larger 
	than~$\pi$. Let~$f$ be a face, let~$e$ and~$e'$ be two consecutive edges on 
	the boundary of~$f$, and let~$\alpha$ be the angle formed by~$e$ and~$e'$ in 
	the interior of~$f$.  If~$e$ and~$e'$ are incident to the root, 
	requirement~\ref{enum:root} implies~$\alpha<\pi$. If both edges contain the
	lowest common ancestor of the leafs belonging to~$f$, then, by 
	requirement~\ref{enum:subtrees}, also~$\alpha<\pi$. In the remaining 
	case,~$e$ and~$e'$ both lie on a path to the left leaf of~$f$, or both lie on 
	a path to the right leaf of~$f$. Let $v$ be the vertex shared by $e$ and $e'$.  At vertex~$v$, we have at least two outgoing 
	edges. Let~$e_1$ be the first outgoing edge and~$e_2$ be the last outgoing 
	edge at~$v$---one of the edges is~$e'$. By the selection of the slope ranks,
	we have $s(e_1)<s(e)<s(e_2)$. Consequently, the supporting line of~$e$ 
	separates~$e_1$ and~$e_2$, and hence both faces containing~$e$ have an angle 
	less than~$\pi$ at~$v$. Therefore, it holds that~$\alpha<\pi$.
   
	Next, we show that the edges and rays of a face do not intersect. Then, by 
	Lemma~\ref{lem:crossing-free}, no edges will cross. Assume that there are two 
	edges/rays~$\ell$ and~$r$ in a common face that intersect in some point~$x$. 
	Let~$t$ be the lowest common ancestor of~$\ell$ and~$r$, and assume 
	that~$\ell$ lies on the path to the left leaf and~$r$ on the path to the 
	right leaf. We define a closed polygonal chain~$P$ as follows. The chain 
	starts with the path~$t\to\ell$, continues via~$x$ to~$r$, and 
	finally returns to~$t$. We direct the edges according to this walk (for 
	measuring the directed slopes) and call them $e_1, e_2,\dots, e_k$. We may 
	assume that~$P$ is simple; otherwise, we find another intersection point.  By 
	Lemma~\ref{lem:inc_slopes}, the slopes are monotone when we traverse~$P$.  
	For~$i=1,\dots,k-1$, let~$\alpha_i$ be the difference between the directed 
	slopes of the edges~$e_i$ and~$e_{i+1}$. Then, the sum~$\sum_{i<k} \alpha_i$ 
	equals the angle between the slopes of~$e_1$ and~$e_k$.  Due to 
	requirement~\ref{enum:subtrees}, this angle is less than~$\pi$. 
	Let~$\beta_i=\pi-\alpha_i$ be the angle between~$e_i$ and~$e_{i+1}$ in~$P$,
	and let~$\beta_0>0$ be the ``interior'' angle at~$t$. We have that 
	$$\sum_{0\le i <k} \beta_i
	= \beta_0 + \sum_{1\le i <k} (\pi -\alpha_i) > 0 + (k-1)\pi - \pi 
	= (k-2) \pi.$$ 
	This, however, contradicts the fact that the angle sum of the polygon with 
	boundary~$P$ is~$(k-2)\pi$. Thus, our assumption that two edges/rays cross was
	invalid.
   
  Since the drawing is assembled from~$n-1$ integral vectors whose absolute coordinates
  are at most~$O(\sqrt{n})$, the complete drawing uses a grid of dimension
  $O(n^{1.5})\times O(n^{1.5})$.
  Since all vectors are reflections of (a subset of) vectors defined by a Farey
  sequence with at most~$n$ entries, Lemma~\ref{lem:primitive_angles}
  yields that the angular resolution is bounded by~$\Omega(1/n)$. 
\end{proof}

We conclude this section with comparing our result with the drawing algorithm 
of Carlson and Eppstein~\cite{ce-tcfoa-GD06}. Their algorithm 
produces a drawing with optimal angular resolution. It draws trees 
convex, but, in contrast to our algorithm, not necessarily strictly
convex.  Allowing parallel leaf edges can have a great impact on the
angular resolution.  However, our ideas can be applied to modify the
algorithm of Carlson and Eppstein.  For the leaf edges, their algorithm uses  
a set of~$k$ slopes and picks the slopes such that they are separated by an 
angle of~$2\pi /k$. The slopes of interior edges have either one of the slopes
of the leaf edges, or are chosen such that they bisect the 
wedge spanned by their outermost child edges. However, it suffices to
assure that the slope of an interior edge differs from the 
extreme slopes in the following subtree by at least~$2 \pi / (2k)$. 

We can now modify the algorithm as follows.
We pick~$2k/8$ primitive vectors and reflect them such that
they fill the whole angular space with~$2k$ distinct integral vectors. 
We use every other vector of this set for the leaf edges.
For an interior edge, we take any vector from our preselected set 
whose slope lies in between the extreme slopes of the edges in its subtree. 
Since we have sufficiently spaced out our set of primitive vectors, we can 
always find such a vector. Thus, we obtain a drawing on the
$O(n^{1.5}) \times O(n^{1.5})$ grid. Clearly, the drawing does not have optimal 
angular resolution. However, since we use~$2k$ integral vectors, which have by 
Lemma~\ref{lem:primitive_angles} an angular resolution of~$\Omega(1/k)$,
we differ from the best possible angular resolution~$2\pi / k$ only by 
a constant factor.  Note that the drawings produced by the algorithm of 
Carlson and Eppstein do not lie on the grid, that is, they do not compute 
rational coordinates for the vertices (by design, since otherwise perfect 
angular resolution cannot be achieved).

\begin{figure}[tb]
  \centering
  \begin{subfigure}[b]{.44\columnwidth}
    \centering
    \includegraphics[page=1]{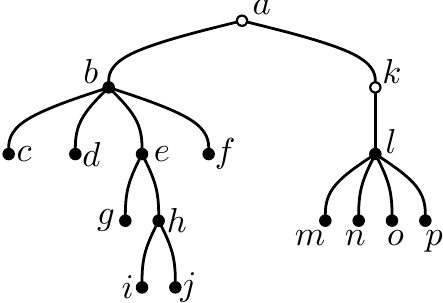}
    \caption{tree with two dummy vertices}
    \label{fig:example-tree-1}
  \end{subfigure}
  \hfill
  \begin{subfigure}[b]{.44\columnwidth}
    \centering
     \includegraphics{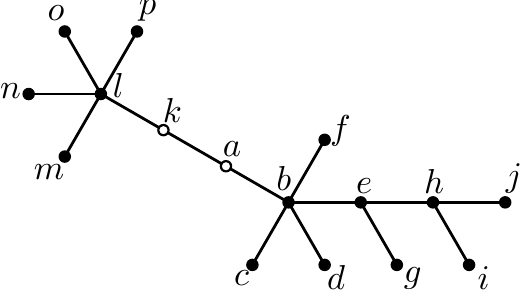}
    \caption{drawing by Carlson and Eppstein~\cite{ce-tcfoa-GD06}}
    \label{fig:example-tree-3}
  \end{subfigure}

  \medskip
  \begin{subfigure}[b]{.44\columnwidth}
    \centering
    \includegraphics[page=3]{angelini-tree}
    \caption{drawing by Angelini et al.~\cite{acbfp-mdg-12}}
    \label{fig:example-tree-2}
  \end{subfigure}
  \hfill
  \begin{subfigure}[b]{.44\columnwidth}
    \centering
    \includegraphics[page=2]{angelini-tree}
    \caption{drawing by our algorithm}
    \label{fig:example-tree-4}
  \end{subfigure}

  \caption{Example tree of Angelini et al.~\cite{acbfp-mdg-12}, drawn by
    various algorithms. We treat the degree-2 vertices as dummy
    vertices because of the degree restriction.}
  \label{fig:example}
\end{figure}

To conclude this section, we present an example that compares our
approach to that of Carlson and Eppstein~\cite{ce-tcfoa-GD06} and to
that of Angelini et al.~\cite{acbfp-mdg-12}; see
Figure~\ref{fig:example}.

\section{Strongly Monotone Drawings}
\label{sec:strong}

In this section, we show how to draw trees in a strongly monotone fashion.
We first show how to draw any proper binary tree, that is, any
internal vertex has exactly two children.  We name our strategy
the \emph{disk algorithm}.  Then, we generalize our
result to arbitrary trees. Further, we show that connected planar
graphs do not necessarily have a strongly monotone drawing.  Finally,
we show how to construct strongly monotone drawings for biconnected outerplanar graphs. 

Let~$T$ be a proper binary tree, let~$D$ be any disk with center~$c$,
and let~$C$ be the boundary of~$D$.  Recall that a strictly convex
drawing cannot have a vertex of degree~2.  Thus, we consider the root
of~$T$ a dummy vertex and ensure that the angle at the root
is~$\pi$.  We draw~$T$ inside~$D$.  We start by  
mapping the root of~$T$ to~$c$.  Then, we draw a horizontal line~$h$
through~$c$ and place the children of the root on~$h \cap
\text{int}(D)$ such that they lie on opposite sites relative to~$c$.
We cut off two circular segments by dissecting~$D$ with two vertical
lines running through points representing the children of the root.
We inductively draw the right subtree of~$T$ into the right circular
segment and the left subtree into the left circular segment.

In any step of the inductive process, we are given a vertex~$v$
of~$T$, its position in~$D$ (which we also denote by~$v$) and a 
circular segment~$D_v$; see Figure~\ref{fig:construction-ind}.  
The preconditions for our construction are that
\begin{enumerate}[label=(\roman*)]
  \item $v$ lies in the relative interior of the chord~$s_v$ that 
    delimits~$D_v$, and
  \item $D_v$ is empty, that is, the interiors of~$D_v$ and~$D_u$ are disjoint 
    for any vertex~$u$ that does not lie on a root--leaf path through~$v$.
\end{enumerate}
In order to place the two children~$l$ and~$r$ of~$v$ (if any), we
shoot a ray~$\vec{v}$ from~$v$ perpendicular to~$s_v$ into~$D_v$.
Let~$v'$ be the point where~$\vec{v}$ hits~$C$.  Consider the chords
that connect the endpoints of~$s_v$ to~$v'$.  The chords and~$s_v$ form
a triangle with height~$vv'$.  The ``height'' is contained in the interior of the 
triangle and splits it into two right subtriangles.  The chords are the
hypotenuses of the subtriangles.  We construct~$l$ and~$r$ by connecting~$v$ to
these chords perpendicularly.  Note that, since the subtriangles are right 
triangles, the heights lie inside the subtriangles.  Hence,~$l$ and~$r$ lie in 
the relative interiors of the chords.  Further, note that the circular 
segments~$D_l$ and~$D_r$ delimited by the two chords are disjoint and
both are contained in~$D_v$.  Hence,~$D_l$ and~$D_r$ are empty, and the
preconditions for applying the above inductive process to~$r$ and~$l$
with~$D_l$ and~$D_r$ are fulfilled.  See Figure~\ref{fig:construction-final} for 
the output of our algorithm for a tree of height~3.

Note that our algorithm does not place the vertices on a grid. However, 
no edge on a strongly monotone path is perpendicular to its monotone direction.
Hence, the vertices can be moved slightly to rational coordinates.
Further, the drawings computed by our algorithms require more than polynomial 
area; in fact, they even require super-exponential area, as the ratio 
between~$s_v$ and~$s_l$ in the inductive step depicted in 
Figure~\ref{fig:construction-ind} cannot be bounded by a constant.
N\"ollenburg et al.~\cite{npr-saicd-GD14} have recently shown that 
exponential area is required for strongly monotone drawings
of trees, which justifies that we cannot produce a drawing on a grid of
polynomial size.

\begin{figure}[tb]
  \begin{subfigure}[b]{.38\columnwidth}
    \centering
    \includegraphics{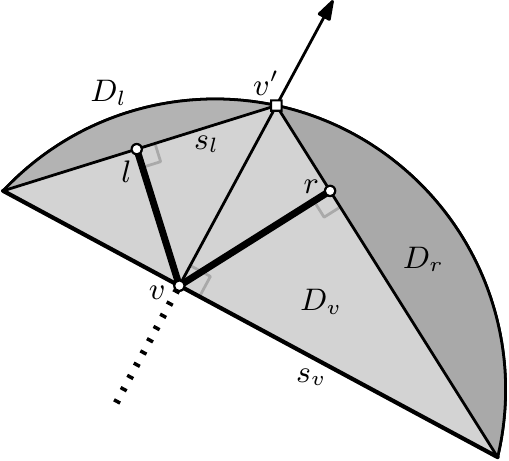}
    \caption{sketch of the inductive step}
    \label{fig:construction-ind}
  \end{subfigure}
  \hfill
  \begin{subfigure}[b]{.58\columnwidth}
    \centering
    \includegraphics{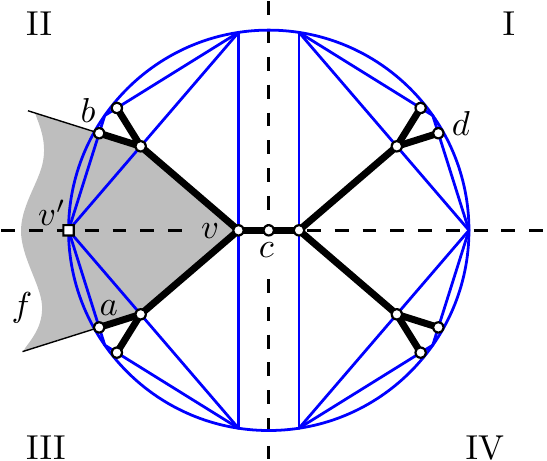}
    \caption{output of our algorithm for a tree of height~3}
    \label{fig:construction-final}
  \end{subfigure}
  \caption{Strongly monotone drawings of proper binary trees}
  \label{fig:construction}
\end{figure}

\begin{lemma}\label{lem:strict-conv-bintree-proof}
  For a proper binary tree rooted in a dummy vertex,
  the disk algorithm yields a strictly convex drawing.
\end{lemma}

\begin{proof}
  Let~$T$ be a proper binary tree and let~$f$ be a face of the
  drawing generated by the algorithm described above.  Clearly,~$f$ is
  unbounded.  Let~$a$ and~$b$ be the leaves of~$T$ that are incident
  to the two unbounded edges of~$f$, and let~$v$ be the lowest common
  ancestor of~$a$ and~$b$; see Figure~\ref{fig:construction-final}.
  Consider the two paths~$v\to a$ and~$v\to b$.  
  We assume that the path from~$v$ through its left child
  ends in~$a$ and the path through its right child ends in~$b$.

  Due to our inductive construction that uses disjoint disk sections
  for different subtrees, it is clear that the two paths do not
  intersect.  Moreover, each vertex on the two paths is convex, that
  is, the angle that such a vertex forms inside~$f$ is less
  than~$\pi$.  This is due to the fact that we always turn right when
  we go from~$v$ to~$a$, and we always turn left when we go to~$b$.
  Vertex~$v$ is also convex since the two edges from~$v$ to its
  children lie in the same half-plane (bounded by~$s_v$).
  
  It remains to show that the two rays~$\vec{a}$ and~$\vec{b}$
  (defined analogously to~$\vec{v}$ above) do not intersect.  To this
  end, recall that~$v'=\vec{v} \cap C$.  By our construction,~$\vec{a}$
  and~$\vec{b}$ are orthogonal to two chords of~$C$ that are
  both incident to~$v'$.  Clearly, the two chords form an angle of
  less than~$\pi$ in~$v'$.  Hence, the two rays diverge, and the
  face~$f$ is strictly convex.
\end{proof}
  
For the proof that the algorithm described above yields a strongly monotone 
drawing, we need the following tools. Let~$\vec v_1$ and~$\vec v_2$ be two 
vectors. We say that~$\vec v_3$ \emph{lies between}~$\vec v_1$ and~$\vec v_2$ 
if~$\vec v_3$ is a positive linear combination of~$\vec v_1$ and~$\vec v_2$.
For two vectors~$\vec v$ and~$\vec w$, we 
define~$\langle \vec v,\vec w\rangle=|\vec v||\vec w|\cos(\vec v,\vec w)$ as the 
scalar product of~$\vec v$ and~$\vec w$.

\begin{lemma}\label{lem:between}
  If a path~$p$ is monotone with respect to two vectors~$\vec v_1$ 
  and~$\vec v_2$, then it is monotone with respect to any vector~$\vec v_3$
  between~$\vec v_1$ and~$\vec v_2$.
\end{lemma}

\begin{proof}
  Let~$\vec v_3=\lambda_1 \vec v_1 + \lambda_2 \vec v_2$ 
  with~$\lambda_1,\lambda_2>0$. Assume that the path~$p$ is given by the 
  sequence of vectors $\vec w_1, \vec w_2, \ldots, \vec w_k$.
  Since~$p$ is monotone with respect to vectors~$\vec v_1$ and~$\vec v_2$, we 
  have that $\langle \vec v_1 , \vec w_i \rangle > 0$ and 
  $\langle \vec v_2 , \vec w_i \rangle > 0$ for all~$i\le k$. This yields, for 
  all~$i\le k$,
  $$ \langle \vec v_3 , \vec w_i \rangle 
  = \langle \lambda_1 \vec v_1 + \lambda_2 \vec v_2 , \vec w_i \rangle 
  =  \lambda_1 \langle \vec v_1 , \vec w_i \rangle  + \lambda_2 \langle \vec v_2 , \vec w_i \rangle 
  > 0, $$
  since~$\lambda_1,\lambda_2 >0$. It follows that~$p$ is monotone with respect 
  to~$\vec v_3$.
\end{proof}

\begin{lemma}\label{lem:strict-mon-bintree-proof}
  For a proper binary tree rooted in a dummy vertex,
  the disk algorithm yields a strongly monotone drawing.
\end{lemma}
  
\begin{proof}
  We split the drawing into four sectors: I, II, III and IV; see
  Figure~\ref{fig:construction-final}. Let~$a$ and~$b$ be two vertices in the 
  graph. We will show that the path~$a\to b$ in the
  output drawing of our algorithm is strongly monotone. 
  Let~$c$ be the root of the tree. Without loss of generality, assume that~$a$ 
  lies in sector III. Then, we distinguish three cases. 
  
  \textbf{Case 1}: $a$ and~$b$ lie on a common root--leaf path; see~$a$ and~$v$ 
    in Figure~\ref{fig:construction-final}. Obviously,~$b$ lies in sector III.
    Without loss of generality, assume that~$b$ lies on the path~$a\to c$.
    By construction, all edges in sector III, seen as vectors
    directed towards~$c$, lie between~$\vec x=(0,1)$ 
    and~$\vec y=(1,0)$. Thus, all edges on the path~$a\to b$, and in
    particular~$\vec{ab}$, lie between~$\vec x$ and~$\vec y$.
    Since~$\vec x$ is perpendicular to~$\vec y$, the path~$a\to b$ is
    monotone with respect to~$\vec x$ and~$\vec y$. Following
    Lemma~\ref{lem:between}, the path $a\to b$ is monotone with
    respect to~$\overrightarrow{ab}$, and thus strongly monotone.
  
  \textbf{Case 2}: $b$ lies in sector I; see~$a$ and~$d$ in
    Figure~\ref{fig:construction-final}. In Case 1, we have shown that
    the all edges on the path~$a\to c$ lie
    between~$\vec x=(0,1)$ and~$\vec y=(1,0)$. Analogously, the same holds
    for the path~$c\to b$. Thus, the path~$a\to b$ is monotone
    with respect to~$\vec x$ and~$\vec y$ and, following
    Lemma~\ref{lem:between}, strongly monotone.
  
    \textbf{Case 3}: $a$ and~$b$ do not lie on a common root--leaf
    path, and~$b$ does not lie in sector~I; see~$a$ and~$b$ in
    Figure~\ref{fig:construction-final}.  Let~$d$ be the lowest
    common ancestor of~$a$ and~$b$. Let~$\langle a=a_0,a_1,\dots,a_k=d\rangle$
    be the path~$d\to a$ to~$a$.
    Further, let~$\langle d=b_0,b_1,\dots,b_m=b\rangle$ be the path~$d\to b$. 
    Finally, let $p=\langle a_k,a_{k-1},\dots, a_0,b_1,\dots,b_{m-1},b_m\rangle$
    be the path~$a\to b$.
    
    \begin{figure}[t]
      \centering
      \includegraphics{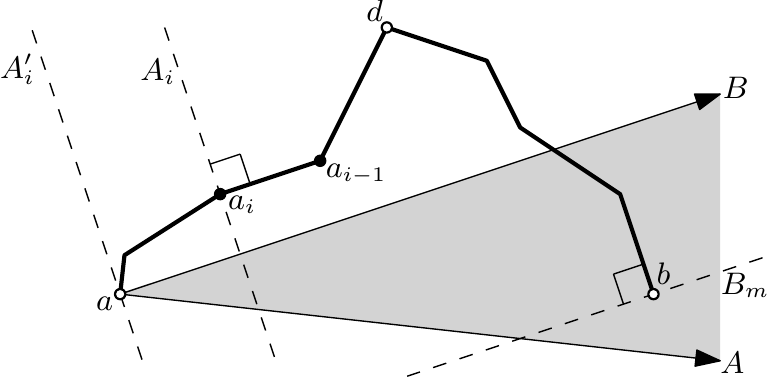}
      \caption{Illustration of case~3 in the proof of
        Lemma~\ref{lem:strict-mon-bintree-proof}}
      \label{fig:strict-mon-bintree-case3}
    \end{figure}
    
    Below, we describe how to rotate and mirror the drawing such that
    any vector~$\overrightarrow{a_ia_{i-1}}$ with $1\le i\le k$ lies
    between~$\vec x=(0,1)$ and~$\vec y=(1,0)$, and any
    vector~$\overrightarrow{b_{j-1}b_j}$ with $1\le j\le m$ lies
    between~$\vec x$ and~$-\vec y$. This statement is equivalent
    to~$x(a_i)<x(d)<x(b_j),1\le i\le k,1\le j\le m$
    and~$y(a_k)<\ldots<y(a_1)<y(d)>y(b_1)>\ldots>y(b_m)$;
    see Figure~\ref{fig:strict-mon-bintree-case3}.
    If~$b$ lies in sector~IV, then~$d=c$ and this statement is true by
    construction. If~$b$ lies in sector~II, then~$d$ is a child
    of~$c$. We rotate the drawing by~$\pi/2$ in counterclockwise
    direction and then mirror it horizontally.
    If~$b$ lies in sector~III, let~$p(d)$ be the parent of~$d$. 
    We rotate the drawing such that the edge~$(p(d),d)$ is drawn vertically. 
    Recall that, by construction, the ray from~$d$ in 
    direction~$\overrightarrow{p(d)d}=-\vec y$ separates the subtrees of the two  
    children of~$d$; see Figure~\ref{fig:construction-ind}. Further, the angle 
    between any edge (directed away from~$d$) in the subtree of~$d$ 
    and~$\overrightarrow{p(d)d}=-\vec y$ is at most~$\pi/2$, that is, they are 
    directed downwards.
   
    For $1 \le i \le k$, let~$A_i$ be the straight line through~$a_i$ perpendicular 
    to~$\overrightarrow{a_{i-1}a_i}$. Let~$A'_i$ be the line parallel to~$A_i$ 
    that passes through~$a$. Due to the $x$-monotonicity of~$p$, the point~$a$ 
    lies below~$A_i$. During the construction of the tree, the line~$A_i$
    defined a circular sector in which the subtree rooted at~$a_i$ including~$a$ 
    was exclusively drawn. It follows that~$a$ and~$b$ lie on opposite sites 
    of~$A_i$. Thus,~$b$ lies above~$A_i$ and also above~$A'_i$. 
    Let~$B_j,1\le j\le m$ be the straight line through~$b_j$ perpendicular 
    to~$\overrightarrow{b_{j-1}b_j}$. Let~$B'_j$ be the parallel line to~$B_j$ 
    that passes through~$a$. By construction,~$b$ lies below~$B_j$ and~$a$ lies
    above~$B_j$. Thus,~$b$ lies below~$B_j'$.
    
    Let~$A$ be the line~$A_i'$ with maximum slope and let~$B$ be the line~$B_j'$
    with minimum slope. First, we will show that the path is monotone with 
    respect to the unit vector~$\vec A$ on~$A$ directed to the right. By our choice 
    of~$A$, the angle between $\vec A$ and any vector~$\overrightarrow{a_ia_{i-1}}$
    with $1\le i\le k$
    is at most~$\pi/2$. Recall that any vector~$\overrightarrow{a_ia_{i-1}}$
    with $1\le i\le k$ lies between~$\vec x$ and~$\vec y$. Since~$\vec A$ is 
    perpendicular to one of these edges and directed to the right, it lies 
    between~$\vec x$ and~$-\vec y$. Since any 
    vector~$\overrightarrow{b_{j-1},b_j}$ with $1\le j\le m$ also lies 
    between~$\vec x$ and~$-\vec y$, the angle between~$\vec A$ and
    any such vector~$\overrightarrow{b_{j-1},b_j}$
    is also at most~$\pi/2$. Thus, the angle between~$\vec A$ and any edge on the
    path~$a\to b$ is at most~$\pi/2$, which shows that the path is monotone with respect 
    to~$\vec A$.
    
    Analogously, it can be shown that the path is monotone with respect 
    to~$\vec B$. Recall that~$b$ lies above~$A$ and below~$B$ and 
    that~$a$ lies above~$B$ and below~$A$.  Hence, the 
    vector~$\overrightarrow{ab}$ lies between~$\vec A$ and~$\vec B$. Following 
    Lemma~\ref{lem:between}, the path is monotone with respect 
    to~$\overrightarrow{ab}$ and, thus, strongly monotone.
\end{proof}

Lemmas~\ref{lem:strict-conv-bintree-proof} 
and~\ref{lem:strict-mon-bintree-proof} immediately imply the following.

\begin{theorem}\label{thm:strict-mon-bintree-proof}
  Any proper binary tree rooted in a dummy vertex
  has a strongly monotone and strictly convex drawing.
\end{theorem}

Next, we (partially) extend this result to arbitrary trees.

\begin{theorem}\label{thm:strict-mon-tree-proof}
  Any tree 
  has a strongly monotone 
  drawing.
\end{theorem}

\begin{proof}
  Let~$T$ be a tree. If~$T$ has a vertex of degree~2, we root~$T$ in this 
  vertex. Otherwise, we subdivide any edge by creating a vertex of degree~2, 
  which we pick as root. Then, we add a leaf to every vertex of degree~2, 
  except the root. Now, let~$v$ be a vertex with out-degree~$k>2$. 
  Let~$(v,w_1),\dots,(v,w_k)$ be the outgoing edges of~$v$ ordered from right to 
  left. We substitute~$v$ by a path~$\langle v=v_1,\dots,v_{k+1}\rangle$, 
  where~$v_{i+1}$ is the left child of~$v_i$, for $i=1,\dots,k$. Then, we 
  substitute the edges~$(v,w_i)$ by~$(v_i,w_i),i=2,\dots,k$; see 
  Figure~\ref{fig:subdivision}.
  
  \begin{figure}[tb]
    \hfill
    \begin{subfigure}[b]{.45\textwidth}
      \centering
      \includegraphics[page=1]{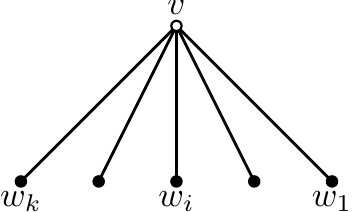}
      \caption{a tree~$T$}
      \label{fig:subdivision-1}
    \end{subfigure}
    \hfill
    \begin{subfigure}[b]{.45\textwidth}
      \centering
      \includegraphics[page=2]{subdivision}
      \caption{the subdivided tree~$T'$}
      \label{fig:subdivision-2}
    \end{subfigure}
    \hfill
    \caption{Subdivision of a vertex~$v$ with~$k$ outgoing edges}
    \label{fig:subdivision}
  \end{figure}
   
  Let~$T'$ be the resulting binary tree.  Clearly, all vertices of~$T'$, 
  except its root, have degree~1 or~3, so~$T'$ is a proper binary tree.
  We use Theorem~\ref{thm:strict-mon-bintree-proof} to get a strongly monotone
  drawing~$\Gamma_{T'}$ of~$T'$. Then, we remove the dummy vertices inserted
  above and draw as straight-line segments the edges of the original
  tree~$T$ that have been substituted or subdivided.
  This yields a drawing~$\Gamma_T$ of~$T$ that is crossing-free
  since the only new edges form a set of stars that are drawn in disjoint
  areas of the drawing.
  
  Now, we show that~$\Gamma_T$ is strongly monotone. Let~$(v,w)$ be an edge
  in~$T$. Let~$p=\langle v=v_1,\dots,v_m=w\rangle$ be the path~$v\to w$ in~$T'$.
  Suppose~$p$ is monotone with respect to some direction~$\vec d$. 
  Thus,~$\angle\{\overrightarrow{v_iv_{i+1}},\vec d\}<\pi/2$ 
  for~$1\le i\le m-1$. 
  Clearly,~$\overrightarrow{vw}=\sum_{i=1}^{m-1}\overrightarrow{v_iv_{i+1}}$
  is a positive linear combination
  of~$\overrightarrow{v_1v_2},\dots,\overrightarrow{v_{m-1}v_m}$
  and, hence,~$\angle\{\overrightarrow{vw},\vec d\}<\pi/2$.
  It follows that the path~$a\to b$ for some vertices~$a,b$ in~$T$ is monotone
  with respect to a direction~$\vec d$ in~$\Gamma_T$ if the path~$a\to b$ is
  monotone to~$\vec d$ in~$\Gamma_{T'}$. With~$\vec d=\overrightarrow{ab}$, 
  it follows that~$\Gamma_T$ is strongly monotone.
\end{proof}

We add to this another positive result concerning biconnected
outerplanar graphs.

\begin{theorem}
  Any biconnected outerplanar graph has a strongly monotone and
  strictly convex drawing on the $O(n^{3/2})\times O(n^{3/2})$ grid.
\end{theorem}

\begin{proof}
  Let~$G$ be a biconnected outerplanar graph with outer cycle
  $\langle v_1,\dots,v_n,v_1\rangle$.
  We place the vertices~$v_1,\dots,v_{n}$
  in order on an $x$- and $y$-monotone convex chain~$C$ that
  has~$v_1$ and~$v_n$ as its endpoints. The chain is assembled by
  translations of $n-1$ primitive vectors, which are sorted by slope
  (see Figure~\ref{fig:biconnected-outerplanar} for a sketch).
  Since the outer cycle is drawn strictly convex, the drawing is planar
  and strictly convex.
  Also, every vector~$\overrightarrow{v_iv_j},j>i$ lies between~$\vec x=(0,1)$
  and~$\vec y=(1,0)$. Thus, by Lemma~\ref{lem:between}, the drawing is
  strongly monotone. 
  
  For our construction we can use a set of
  primitive vectors whose coordinates are bounded by $O(\sqrt{n})$.
  Since we have linked $n$ such vectors, the
  asserted bound follows.
\end{proof}

We close with a negative result.  Angelini et
al.~\cite[p.~33]{acbfp-mdg-12} stated that they are ``not aware of any
planar graph not admitting a planar monotone drawing for any of its
embeddings''.  We give the first family of graphs that
  do not admit any \emph{strongly} monotone drawing.  Note that the
graphs in the family that we construct are neither outerplanar nor
biconnected.

\begin{theorem}
  There is an infinite family of connected planar graphs
  that do not have a strongly monotone drawing in any combinatorial embedding.
\end{theorem}

\begin{proof}
  Let~$C$ be the graph that arises by attaching to each vertex
  of~$K_4$ a ``leaf''; see Figure~\ref{fig:nostrict-connected}.
    \begin{figure}[tb]
    \begin{minipage}[t]{.5\columnwidth}
      \centering
      \includegraphics{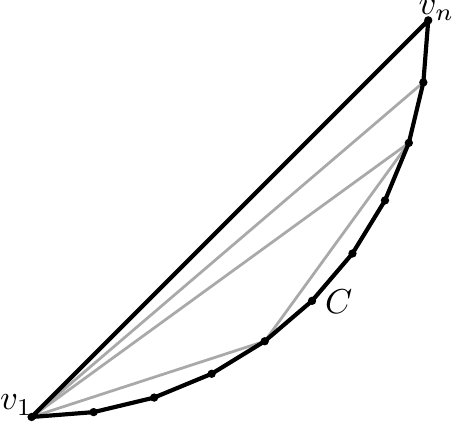}
      \caption{A strongly monotone drawing of a biconnected outerplanar graph}
      \label{fig:biconnected-outerplanar}
    \end{minipage}
    \hfill
    \begin{minipage}[t]{.43\columnwidth}
      \centering
      \includegraphics{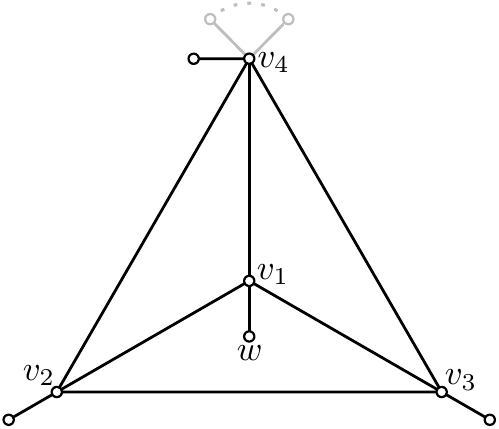}
      \caption{A planar graph without any strongly monotone drawing}
      \label{fig:nostrict-connected}
    \end{minipage}
  \end{figure}
  Let~$v_1,\dots,v_4$ be the vertices of~$K_4$. For~$K_4$ to be
  crossing-free, one of its vertices, say~$v_1$, lies in the interior.
  Let~$w$ be the leaf incident to~$v_1$. 
  Because of planarity,~$w$ has to be placed inside
  a triangular face incident to~$v_1$. Without loss of generality, assume
  that~$w$ is placed in the face~$(v_1,v_2,v_3)$. If the drawing is
  strongly monotone, 
  then~$\angle(\overrightarrow{wv_2},\overrightarrow{wv_1})<\pi/2$
  and~$\angle(\overrightarrow{wv_1},\overrightarrow{wv_3})<\pi/2$,
  and thus~$\angle(\overrightarrow{wv_3},\overrightarrow{wv_2})>\pi$.
  However, this means that~$w$ does not lie inside the triangle~$(v_1,v_2,v_3)$,
  which is a contradiction to the assumption. Thus,~$C$ does not have
  a strongly monotone drawing in any combinatorial embedding.
  We create an infinite family from~$C$ by adding more leaves to the
  vertices of~$K_4$. 
\end{proof}

\section{Conclusion}
\label{sec:open}

We have shown that any tree has a convex monotone drawing on a grid with area
$O(n^3)$ and a strongly monotone drawing, but several problems remain open.
It is an open question whether any tree has a strongly monotone drawing on a grid 
of exponential size. We have shown that not every connected planar graph admits a strongly
monotone drawing, while Felsner et al.~\cite{fikkms-smdpg-socg16} showed that every
triconnected planar graph does so.
It is still open whether there is a biconnected planar 
graph that does not have any strongly monotone drawing. If yes, it is
interesting whether this can be tested efficiently. 

\bibliographystyle{abbrvurl}
\bibliography{abbrv,monotone}

\end{document}